\documentclass[runningheads]{llncs}
\usepackage{amsmath,amssymb}
\usepackage{tikz}
\usepackage{hyperref}
\usepackage{livesynthesis}
\usepackage{ltl}
\usepackage{graphicx}

\usepackage{subcaption} 
\usepackage[strings]{underscore}	 
\usepackage{pifont}
\usepackage{cleveref}
\usepackage{todonotes}
\usepackage{xspace}
\usepackage{ntheorem}
\theoremstyle{plain}

\usetikzlibrary{calc}
\newtheorem*{theorem*}{Theorem}

\newcommand{\twopartdef}[4]
{
	\left\{
		\begin{array}{ll}
			#1 & \mbox{if } #2 \\
			#3 & \mbox{if } #4
		\end{array}
	\right.
}

\begin{document}

\title{Live Synthesis\thanks{This work was partially supported by the German Research Foundation (DFG) as part of the Collaborative Research Center ``Foundations of Perspicuous Software Systems'' (TRR 248, 389792660), by the European Research Council (ERC) Grant OSARES (No. 683300), and by the German Israeli Foundation (GIF) Grant ``Knowledge-based Synthesis'' (No. I-1513-407./2019).}}
\titlerunning{Live Synthesis}

\author{Bernd\ Finkbeiner \and Felix Klein\ \and Niklas\ Metzger }
\authorrunning{B. Finkbeiner\ \and F. Klein\ \and N.\ Metzger}

\institute{
CISPA Helmholtz Center for Information Security, Saarbr\"ucken, Germany\\
\email{\{finkbeiner,felix.klein,niklas.metzger\}@cispa.de}\\
}

\maketitle

\begin{abstract}
Synthesis automatically constructs an implementation that satisfies a
given logical specification.  In this paper, we study the \emph{live
synthesis} problem, where the synthesized implementation replaces an
already running system. In addition to satisfying its own specification,
the synthesized implementation must guarantee a sound transition
from the previous implementation. This version of the synthesis problem is
highly relevant in ``always-on'' applications, where updates 
happen while the system is running.
To specify the correct handover between the old and new
implementation, we introduce an extension of linear-time temporal
logic (LTL) called \emph{\liveltl}.  A \liveltl specification defines separate
requirements on the two implementations and ensures
that the new implementation satisfies, in addition to its own
requirements, any obligations left unfinished by the old
implementation. For specifications in \liveltl, we show that the live
synthesis problem can be solved within the same complexity bound as
standard reactive synthesis, i.e., in 2EXPTIME. 
Our experiments show the necessity of live synthesis for \liveltl 
specifications created from benchmarks of SYNTCOMP and robot control.

\end{abstract}

\section{Introduction}\label{sec:introduction}

The past decade has brought remarkable progress in the automatic synthesis of reactive systems from temporal specifications~\cite{SYNTCOMP,bosy,strix}. Traditionally, synthesis is seen as a  one-off method: the generated implementation is guaranteed, by construction, to satisfy the specification. If the specification changes, the process is repeated from the start. For systems that are \textit{always-on}, like banking systems, or controllers in power plants, this may, however, not be an option: when the requirements change, the system must be updated while it is still running, and the control must transition to the new version without disrupting the safety or functionality of the running system. While such \emph{live updates} are a well-studied concern in operating systems research (cf. \cite{operatingsystems}), they are, somewhat surprisingly, still a novelty in formal methods.

In this paper, we define a \textit{live system} as sequence of  implementations, each with a corresponding specification. The last element in the sequence is the currently executed system. 
Performing a live update terminates the currently active system and extends the sequence with a new implementation.
The key challenge of live updates is that 
any obligations imposed by the specification of the terminated system that are not yet satisfied at the time of the update must be taken care of by the newly active system.
This transfer of obligations is important to make the update transparent from the  user's perspective.
Consider, for example, an arbiter specified as the LTL formula $\LTLglobally(\mathit{request} \rightarrow \LTLfinally \mathit{grant})$, which requires that every $\mathit{request}$ is eventually followed by a $\mathit{grant}$. If the update occurs after some $\mathit{request}$, but before the corresponding $\mathit{grant}$, then the new implementation must still guarantee the occurrence of the $\mathit{grant}$.

The problem of \emph{model checking} live updates is to check whether a given new implementation will result in a correct live update; the \emph{synthesis} problem is to automatically find such an implementation.
To specify the correct handover between the old and new implementation, we introduce an extension of linear-time temporal logic (LTL) called \emph{LiveLTL}. A LiveLTL specification defines requirements on the two implementations and ensures that the new implementation satisfies, in addition to its own requirements, any obligations left unfinished by the old implementation. 
We consider two variants of the model checking and synthesis problems. In \emph{finite-trace live updates}, we only require the update to be correct in a specific situation, i.e., after a specific  execution of the previous implementation. In \emph{universal updates}, we require that the update can occur at any time. 
We show that model checking live updates is \textsc{PSPACE}-complete in the initial and update specification.
Synthesis is 2\textsc{EXPTIME}-complete in the combination of the specifications for both update variants.

We report on experience with a prototype implementation of our approach on a range of benchmarks, including examples taken from the synthesis competition and a robotic case study. 
In our experiments, live synthesis is used to construct live updates built on reasonable pairs of specifications.
The results show the necessity of verifying live updates with the adapted semantics of \liveltl and that every considered specification states obligations for the update.

\section{Running Example -- Relay Station}\label{sec:example}
Consider the following setup: a satellite has been positioned in the
orbit of Mars in combination with multiple base stations on the
planet. The base stations take samples from the extraterrestrial
environment, analyze them and submit their findings to the
satellite. After the data has been sent by a station, it waits for
instructions from the satellite: whether the sample must be
further analysed, or whether it can be discarded and a new sample must
be taken. The satellite, on the other hand, provides the stations with
the corresponding instructions and collects the data of all stations
for relaying it back to earth. To this end, the satellite takes care
that always some data of all base stations has been collected to be
present in the report for earth.

We formalize this behavior of the satellite in LTL. On the input side, the satellite receives $ n $
\textit{measurements}~$ m_{j} $ of every base station, where
$ 0 \leq j < n $ ranges over the $ n $ deployed base stations on the
planet. On the output side, the satellite outputs
\textit{instructions}~$ i_{j} $ and can create a \textit{report}~$ r $
to be sent back to earth. The behavior is formalized using the
following guarantees:
First of all, every measurement~$ m_{i} $ must be responded to
eventually and instructions are only sent in response to received measurements.
\vspace{-.2cm}
\begin{equation*}
  \phi_{1} ~:=~ \bigwedge\limits_{j=0}^{n} m_{j} \rightarrow\, \LTLnext \LTLeventually i_{j} \quad\quad \phi_{2} ~:=~ \bigwedge\limits_{j=0}^{n} \LTLglobally \neg m_{j} \rightarrow\, \LTLeventually \LTLglobally \neg i_{j}
\vspace{-.2cm}
\end{equation*}
Furthermore, a report is generated as long as every base station
submits a measurement regularly, while no report needs to be generated as long as some measurements
are still missing.
\vspace{-.2cm}
\begin{equation*}
  \phi_{3} ~:=~ (\bigwedge\limits_{j=0}^{n} \LTLeventually m_{j}) \rightarrow\, \LTLeventually r \quad\quad \phi_{4} := (\bigvee\limits_{j=0}^{n} \LTLglobally \neg m_{j}) \rightarrow\, \LTLeventually \LTLglobally \neg r
\vspace{-.2cm}
\end{equation*}
All guarantees $ \phi_{j} $ must be satisfied at every point in
time. We obtain the overall specification
$ \varphi := \bigwedge_{j=1}^{4} \LTLglobally \phi_{j} $. The
specification is realizable, as witnessed by the synthesized labeled transition system (LTS) for
two base stations in \Cref{fig:relaystation}.
We follow the transition system for $\LTLglobally \phi_{1}$. 
Starting in the initial state, if $m_0$ and $m_1$ is received, we stay in the same state and $\LTLnext \LTLfinally m_0$ as well as $\LTLnext \LTLfinally m_1$ is satisfied.
The transition system follows the $\neg m_0$ edge to the state labeled with $i_1$ to satisfy the subformula $\LTLnext \LTLfinally i_1$.
Note that $m_1 \rightarrow \LTLeventually i_1$ would be directly satisfied in the initial state since the Moore semantics evaluates the formula based on the current state and next edge label.
The states at the top right and bottom left ensure that $\varphi_2$ is satisfied, i.e., it waits for inputs before the corresponding output is set.
Corresponding to $\varphi_4$, the top left and bottom right states control the output $r$ which is only allowed to be true as long as all measurements are received.
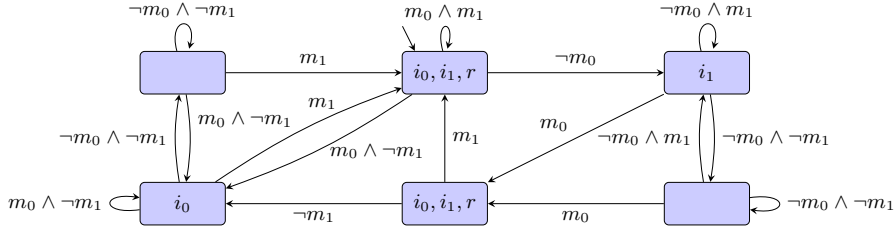
\begin{figure}[t]
  \centering
\resizebox{.99\textwidth}{!}{
  \begin{tikzpicture}

    \node[draw,fill=blue!20,minimum height=2em,minimum width=4em,rounded corners=2] (N0) at (0,0) {$ i_{0}, i_{1}, r $};        
    \node[draw,fill=blue!20,minimum height=2em,minimum width=4em,rounded corners=2] (N1) at (4,0) {$ i_{1} $};
    \node[draw,fill=blue!20,minimum height=2em,minimum width=4em,rounded corners=2] (N2) at (4,-2) {$ $}; 
    \node[draw,fill=blue!20,minimum height=2em,minimum width=4em,rounded corners=2] (N3) at (0,-2) {$ i_{0}, i_{1}, r $};
    \node[draw,fill=blue!20,minimum height=2em,minimum width=4em,rounded corners=2] (N4) at (-4,-2) {$ i_{0} $}; 
    \node[draw,fill=blue!20,minimum height=2em,minimum width=4em,rounded corners=2] (N5) at (-4,0) {$ $};

    \node[minimum height=0.em,minimum width=2.1em,anchor=north] (NN) at (N0.north) {};
    \node[minimum height=0.2em,minimum width=2.1em,anchor=east] (NX) at (N3.east) {};
    \node[circle,inner sep=0pt] (I) at ($ (N0.north west) + (0,0.4) $) {};

    \path[->,>=stealth]
    (I)
      edge ($ (N0.north west) + (0.2,0) $)
    (NN)
      edge[loop above] node {$ m_{0} \wedge m_{1} $} (N0)
    (NX)
      edge[loop right,opacity=0] node[opacity=0] {\phantom{$ m_{0} \wedge m_{1} $}} (N3)            
    (N0)
      edge node[above] {$ \neg m_{0} $} (N1)
      edge[bend left=7] node[below right,yshift=5] {$ m_{0} \wedge \neg m_{1} $} (N4)
    (N1)
      edge[loop above] node {$ \neg m_{0} \wedge m_{1} $} (N1)
      edge[bend left=9] node[right] {$ \neg m_{0} \wedge \neg m_{1} $} (N2)
      edge node[above left] {$ m_{0} $} (N3)
    (N2)
      edge[loop right, in=352, out=8, min distance = 6mm] node {$ \neg m_{0} \wedge \neg m_{1} $} (N2)
      edge[bend left=9] node[left] {$ \neg m_{0} \wedge m_{1} $} (N1)
      edge node[below] {$ m_{0} $} (N3)
    (N3)
      edge node[right] {$ m_{1} $} (N0)
      edge node[below] {$ \neg m_{1} $} (N4)
    (N4)
      edge[loop left, out=188, in=172, min distance = 6mm] node {$ m_{0} \wedge \neg m_{1} $} (N4)
      edge[bend left=7] node[above right, yshift=5, xshift=-2] {$ m_{1} $} (N0)
      edge[bend left=9] node[left] {$ \neg m_{0} \wedge \neg m_{1} $} (N5)
    (N5)
      edge[loop above] node {$ \neg m_{0} \wedge \neg m_{1} $} (N5)
      edge[bend left=9] node[right, yshift=8] {$ m_{0} \wedge \neg m_{1} $} (N4)
      edge node[above] {$ m_{1} $} (N0)
    ;
  \end{tikzpicture}}
  \caption{Synthesized LTS for the satellite specification.}
  \label{fig:relaystation}
\end{figure}%
Consider a situation, where one of the base stations fails. 
The satellite controller must be updated, since the satellite would
wait indefinitely for the data of the broken base station
otherwise. The report generation would also be broken. However, we
cannot just eliminate the broken base station from the original
specification, synthesize again and restart the satellite with the new
result. The reason is that there still may be an outstanding
instruction of the satellite for one of the remaining base stations,
for which this base station is actively waiting. Therefore, the
updated specification still needs to take this obligation of the old
implementation into account.

In the remainder of this paper we consider the necessary changes to the
synthesis procedure that are required for a correct update of the
specification and synthesized implementation.
An adapted verification framework is introduced that enables the validation of live systems.
We present a logic that avoids the break of the base stations and satellite due to the disregarded obligations of the old system during update.

\section{Preliminaries}\label{sec:prelimnaries}
\emph{Linear Temporal Logic.} %
Linear temporal logic (LTL) \cite{TheTemporalLogicsOfPrograms} is a logic for specifying correctness of linear-time systems.
The syntax is a combination of state and path operators over a set of atomic propositions (AP) that define behavior over infinite time.
Formulas in LTL are built according to the grammar 
$\phi ::= \LTLtrue \mid \LTLfalse \mid a \mid \neg \phi \mid \phi_1 \wedge \phi_2 \mid \LTLnext \phi \mid  \phi_1 \LTLuntil \phi_2$
where $a \in AP$. 
Temporal operators are \textit{next} $\LTLnext$ and \textit{until} $\LTLuntil$, all other operators are boolean connectives.
We assume every LTL formula to be in release positive normal form (PNF) where negations are only allowed in front of atomic propositions.
For readability, implication $\rightarrow$ and equivalence $\leftrightarrow$ as well as the common abbreviations \textit{eventually} $\LTLeventually a $ for $\LTLtrue \LTLuntil a$ and \textit{globally} $\LTLglobally a$ for $\neg\LTLeventually \neg a$ are used throughout this paper.
Defining the LTL semantics, the operator $\vDash$ evaluates infinite traces $\sigma$ and explicit index $i$ against LTL formulas~$ \varphi $ where traces are words over letters $ \sigma \in \apomega$. 
For example, $\sigma$ satisfies $\LTLnext a$ if in the next step $a$ holds in $\sigma$ and $a \LTLuntil b$ if $a$ holds until $b$ holds.
A trace $\sigma = A_0A_1A_2 \ldots$ with $A_i \in 2^{AP}$ is an infinite sequence of sets of atomic propositions.
We use the infix notation $\sigma[n,m]$ to crop the trace to the sub-trace from position n to m, $\sigma[n,m] = A_nA_{n+1}\ldots A_{m-1}$, where $A_i \in 2^{AP}$, and concatenate the finite trace $\sigma_1$ with the possibly infinite trace $\sigma_2$ with $\sigma_1 \cdot \sigma_2$.
The semantic operator $\vDash$ builds a language of a specification $\phi$ with $\text{Words}(\phi) = \{\sigma \in \apomega \mid \sigma, 0 \vDash \phi\}$.
A trace $\sigma$ that is terminated at an arbitrary position $m$, i.e., $\sigma[0,m]$, is a finite trace and denoted by $\trace$.
The function $expand: LTL \rightarrow LTL$ uses the standard LTL expansion rules to unroll the given formula, $expand_n$ repeats $expand$ $n$ times. For example, $expand_1(a U b) = b \vee (a \wedge \LTLnext (a \LTLuntil b))$.
The function $\af : LTL \times 2^{AP} \rightarrow LTL$ \cite{FromLTLToDeterministicAutomata} evaluates the formula on a given atomic proposition assignment and returns the remaining formula, e.g. $\af(a \LTLuntil b, \{a\}) = a \LTLuntil b$ and $\af(a \LTLuntil b, \{b\}) = \LTLtrue$.
$\af(\phi, \sigma[0,n])$ is defined as $\af(\af(\phi, \sigma_0), \sigma[1,n])$ with $\af(\phi, \epsilon) = \phi$. 

\noindent
\emph{Transition Systems.}
The reactive model for LTL are transition systems where state labels correspond to the output of systems and transition labels correspond to the input of the environment.
Given a finite set of directions $\Upsilon$ and a finite set of labels $\Sigma$, a $\Sigma$-labeled $\Upsilon$-transition system is a tuple $TS = (T, t_0, \tau, o)$, consisting of a finite set of states $T$, an initial state $t_0 \in T$, a transition function $\tau: T\times \Upsilon \rightarrow T$, and a labeling function $o: T \rightarrow \Sigma$.  
Given $AP$ and partition $AP = O \cup I$ for output and input atomic propositions, implementations for LTL specifications are $2^O$-labeled $2^I$-transition systems ($TS$). 
The paths of a transition system start in $t_0$ and follow the transition function $\tau$ collecting input and output labels with the output function $o$. 
The traces of a transition system $\mathit{Traces}(TS)$ omit the state information of paths.
We assume transition systems without terminal states and a deterministic transition function.

\noindent
\emph{Model Checking and Synthesis.} Model checking a transition system $TS$ against a specification $\phi$ checks the relation $\mathit{Traces}(TS) \subseteq \mathit{Words}(\phi)$.
The problem of automatically constructing a transition system that satisfies the model checking property is referred to as \textit{synthesis}.
In the course of this paper, we refer to the algorithms of LTL model checking and synthesis as black box algorithms.
Similar to $\mathit{Traces}(TS)$, we denote the set of finite traces of $TS$ by $\mathit{FinTraces}(TS)$.

\section{Live Updates}\label{sec:liveupdates}
Common formalisms for verification agree on the following assumption: different system versions are analyzed in isolation, i.e., everything that happened before the initial state of the new implementation is irrelevant for its correctness.
For updates at runtime, this assumption is infeasible.
The update system has to satisfy \textit{obligations} that were stated during the execution of the previous system to be correct.
In this section, we set the foundations for a specification language that is able to express correctness of a live update by defining the structure of two live update problems.
We identify the factors affecting the update process and formalize the interplay of the components.
The definitions are independent of specific temporal logics and can be adapted to various logics and system models.

Proving the correctness of systems either by model checking or synthesis assumes the existence of a starting point that is handled as the initial state.
For live updates, the starting point of verification is not the initial state of the update system, but the initial state of the system running beforehand.
Running systems create obligations that cannot be discarded when updated live, otherwise, for example, an observer would starve waiting for its response.
The recent development of live systems enforces the sensibility of correctness algorithms to validate systems w.r.t.\ the \textit{context} they are started in.
For linear-time systems given as transition systems, we define the context as the finite execution of the previous system combined with its specification.
The finite execution implicitly changes the state of the formula which we refer to as \textit{active} formula.
We capture this change to the formula with a function $\Psi$, which, given a finite trace and a specification, returns a specification that captures the obligations needed for the satisfaction of the update system. 
With defining $\Psi$, one is able to vary the impact of the initial system to the update system.
Verifying an update system with standard LTL, one implicitly defines $\Psi$ to be $\LTLtrue$ for every input, enforcing no obligations on the update system.
\begin{definition}[Finite Trace Live Update]\label{def:finitetraceliveupdate}
Let~$TS_I$ be an initial system, $TS_U$ be an update system, $\phi$ be an initial specification, $\psi$ be an update specification, and $\eta$ be a finite trace of $TS_I$.
$TS_U$ is considered correct if it is correct w.r.t.\ $\psi$ and the result of $\Psi(\trace,\phi)$ for the function $\evolve : (2^{AP})^{*} \times LTL \rightarrow LTL$ defining the obligation.
\end{definition}
\begin{figure}[t]
\centering
\resizebox{.98\textwidth}{!}{
\centering
    \resizebox{\textwidth}{!}{
    \begin{tikzpicture}[every text node part/.style={align=center},x = 6cm]
        \node[draw, rectangle]  (4) at (0.25,0) {\textit{initial system} \\ $TS_I \vDash \phi$}; 
        \node[draw,  rectangle] (3) at (1,0){\textit{finite execution}\\$\trace_{[0,n-1]} \in FinTraces(TS_I)$};
        \node[draw,  rectangle] (1) at (2,.5){\textit{update system}\\$TS_U$};
        \node[draw= none, rectangle] (2)  at (1,0.9){};
        \node (node) (0) at (2.5, .9) {\cmark};
        \node (node) (00) at (2.5, .2) {\xmark};
        \draw[->, thick]
        (4) edge node[below, sloped]  {\textit{n-steps}} (3)
        (3) edge node [below, sloped] {$\evolve(\trace_{[0,n-1]}, \phi)$}(1)
        (2) edge node [above, sloped] {\textit{update specification  }$\psi$}(1)
        (1) edge (0)
        (1) edge (00)
        ; 
    \end{tikzpicture}}
}
\caption{The finite trace live update with $\varphi$ as the initial specification, $\psi$ as the update specification, and $\Psi$ as the function computing the obligation for $TS_U$.}
\label{fig:liveupdates}
\end{figure}
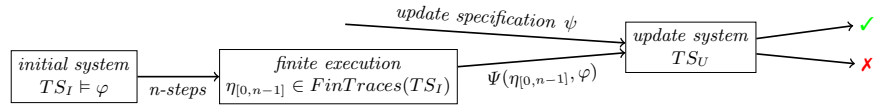
The finite trace live update handles the context of the update as white-box: the finite execution of the previous system is fully known.
For this explicit execution, the obligation is computed and, together with the specification of the update system, verified against the update.
\Cref{fig:liveupdates} shows the dependencies built by the finite trace live update where $n$ is the number of discrete time-steps of the finite execution.
However, the explicit finite execution of the initial system is not always available. 
Therefore, \Cref{def:universalliveupdate} introduces update correctness for all possible finite paths of the initial system.
\begin{definition}[Universal Live Update]\label{def:universalliveupdate}
Let~$TS_I$ be an initial system, $TS_U$ be an update system, $\phi$ be an initial specification, and $\psi$ be an update specification. $TS_{U}$ is considered correct if it is correct w.r.t.\ $\psi$ and $\evolve(\eta, \phi)$ for all possible finite traces $\eta$ of $TS_I$.
\end{definition}
The context of the update is handled as black-box in the universal case. 
The explicit execution and the system's state of the update is unknown.
Nevertheless, if all possible obligations are satisfied by the update system, the update is guaranteed to be correct.
\Cref{def:universalliveupdate} increases the number of possibilities to be verified, since arguing over an infinite set of finite traces cannot be performed directly.
In comparison to the explicit live update, the length $n$ is kept arbitrary since every finite trace may enforce its particular obligation.

Since we consider reactive systems, it is natural to aim for an update system that reacts to the update context and contains different initial states for different contexts, i.e., for each result of $\Psi(\trace, \phi)$ the update system starts differently.
Note that this problem is covered by the finite trace live update if the number of different contexts is finite.
One can solve the update problem for each context and combine the resulting update systems accordingly.
In general, multiple other meaningful models of update correctness can be designed, e.g., an existential version defining the existence of an update point in the initial system's future.
Nevertheless, finite trace and universal live updates suffice for the course of this paper and build a justifiable framework for live updates.

\section{A Temporal Language for Live Updates}\label{sec:liveltl}
With the two live update problems defined, we introduce \liveltl to state and verify the correctness of live updates.
\liveltl is an extension to LTL and specifies live update properties that automatically enforce the obligations of the previous execution on the update system.
The syntax and semantics of \liveltl as well as the language equivalence to LTL are shown.
Moreover, we identify the class of obligations that can be stated by \liveltl specifications.

\subsection{LiveLTL}
\liveltl is designed according to three aspects: (1) the initial system is not able to enforce new obligations after termination, (2) all obligations stated before termination are satisfied by the update system, and (3) obligations are satisfiable in finite time.
This guideline is a trade-off between independence of the previous system and incurring obligations from the initial specification to the update system.
The definition of \liveltl follows the finite trace update structure and builds the language for inputs as a combination of a finite and an infinite trace evaluation.
The syntax is taken from LTL and we assume the set of atomic proposition for the initial system to be a subset of the atomic propositions of the update system.
As extension to the semantic operator $\vDash$ of LTL, the operators $\vDashInitial$ and $\vDashUpdate$ form the language for the initial system and the update system respectively.
$\vDashUpdate$ performs an index shift from time-step 0 to the update position and evaluates the changed formula with the LTL operator and is defined as $\sigma, i  \vDashUpdate \phi \text{ iff }\ \sigma, i+|\trace| \vDash \phi$.
Since the update specification is only relevant for the update system, the shift of size $|\trace|$ enables the correct evaluation of the update system's part of the trace.
$\vDashInitial$ inserts $|\trace|$ as upper bound for recurrent formulas, i.e., formulas with the \textit{release} operator:
\begin{align*}
\sigma, i & \vDashOriginal  \LTLtrue && \phantom{t}\sigma, i  \nvDashOriginal  \LTLfalse\\
    \sigma, i & \vDashOriginal  a && \text{ iff }\ A_i \vDash a, \text{i.e. } a \in A_i\\
    \sigma, i & \vDashOriginal  \neg a && \text{ iff }\ A_i \nvDash a, \text{i.e. } a \notin A_i\\
    \sigma, i  & \vDashOriginal \varphi_1 \wedge \varphi_2  &&\text{ iff }\ \sigma, i \vDashOriginal \varphi_1 \text{ and}\ \sigma, i \vDashOriginal \varphi_2\\
    \sigma, i  & \vDashOriginal \varphi_1 \vee \varphi_2  &&\text{ iff }\ \sigma, i \vDashOriginal \varphi_1 \text{ or}\ \sigma, i \vDashOriginal \varphi_2\\
    \sigma, i & \vDashOriginal \LTLnext \varphi &&\text{ iff }\ {\sigma, i+1 \vDashInitial \varphi}\\
    \sigma, i & \vDashOriginal \varphi_1 \LTLuntil \varphi_2 && \text{ iff }\ \exists j, j \geq i.\ \sigma,j \vDashOriginal \varphi_2\ \text{and} \ \forall k,i\leq k < j.\ \sigma, k \vDashOriginal \varphi_1\\
    \sigma, i & \vDashOriginal \phi_1 \LTLrelease \phi_2 && \text{ iff }\ \forall j,{ \color{red}|\trace| > j \geq i.}\ \sigma, j \vDashOriginal \phi_2\ \text{or}\\
    &&&\phantom{ a}\ \exists k,{ \color{red}|\trace| > k \geq i.}\ (\sigma, k \vDashOriginal \phi_1\wedge \forall l, i \leq l \leq k.\ \sigma, l \vDashOriginal \phi_2)
\end{align*}
Informally, $\phi_1 \LTLrelease \phi_2$ opens the \textit{obligation} $\phi_2$ in every execution step which contradicts (1) if evaluated after the update.
As standard LTL semantics enables the specification to infinitely open new obligations, $\vDashInitial$ is built to limit this behavior to the actual finite execution of the initial system.
The definition of $\vDashInitial$ mostly follows the definition of $\vDash$, except for the evaluation of \textit{release} formulas.
For all indices greater or equal to the length of the trace, $\phi_1 \LTLrelease \phi_2$ is immediately satisfied, thus imposing the end of newly created obligations from the initial implementation.
Therefore, the initial operator permits the transfer of finitely satisfiable obligations to the update system (2), but forbids the impact of the initial system after its termination (1).
Note that for LTL formulas in PNF, all operators except \textit{release} only specify finite behavior and all open obligations are satisfiable in finite time (3).
The newly introduced operators are used to define the language of \liveltl.
\begin{definition}[Language of \liveltl]
    Let $\phi,\psi$ be LTL formulas and let $\trace \in (2^{AP})^{*}$. The linear time property induced by $\phi, \psi$, and $\trace$ is
    \vspace{-0.12cm}
     \[\text{Words}(\phi,\psi, \trace) = \{\trace \cdot \sigma \in (2^{AP})^{\omega} \mid
     \trace \cdot \sigma, 0 \vDashInitial \phi\ \wedge\trace \cdot \sigma, 0 \vDashUpdate \psi\}.\]
\end{definition}
     \vspace{-0.09cm}
The language is dependent on the initial specification, the update specification, and the finite trace.
Evaluating the inclusion of an infinite trace with the first $|\trace|$ elements being fixed consists of a combination of the operators $\vDashInitial$ and $\vDashUpdate$.
The initial \liveltl operator is defined on the syntactic structure of the initial formula and is insensitive with respect to syntactic tautologies.
Providing formulas without syntactic ambiguity that cannot be dissolved in $|\trace|$ time steps is left to the specifier.
The following theorem relates \liveltl and LTL.
\begin{theorem}\label{th:liveltlltl}
    \liveltl and LTL are equally expressive.
\end{theorem}
The proof is a reduction via encoding the initial trace into the LTL formula.
While being equally expressive, \liveltl enables the direct evaluation of the newly introduced live update problems on a given context.
Correctness for finite trace live updates follows from standard language inclusion.
\begin{definition}[Finite Trace LiveLTL Update]\label{def:finitetraceliveupdateltl}
    Let $TS_U$ be an update system, $\phi$ be an initial specificaiton, $\psi$ be an update specification, and $\eta$ be a finite trace.
    $TS_U$ is correct w.r.t. finite trace \liveltl if
    $\trace \cdot \text{Traces}(TS_U) \subseteq \text{Words}(\phi, \psi, \trace).$
\end{definition}
\begin{example}
Interpreting the running example as finite trace LiveLTL update, we can obtain the finite trace $\trace = \{m_1, i_0, i_1, r\}, \{i_1\}, \{m_0, m_1\}$ as execution of the relay station.
Evaluating $\trace$ with $\vDashInitial$ shows that $\LTLfinally\,i_0$, $\LTLfinally\,i_1$, and $\LTLfinally\,r$ need to be satisfied by the update system, since both measurements are unanswered and no report was given after both base stations sent their measurements.
Note that changing the last trace element to $\{m_0\}$ eliminates the obligations for the base station $i_1$  and the report $r$.
\end{example}
The finite trace update directly translates to the definition of \liveltl, whereas the universal live update adds a level of quantification.
\begin{definition}[Universal Live LTL Update]\label{def:universalliveupdateltl}
    Let $TS_I$ be an initial system, $TS_U$ be an update system, $\phi$ be an initial specification, and $\psi$ be an update specification.
   	$TS_{U}$ is correct w.r.t. universal \liveltl if
   	\vspace{-0.1cm}
    \[\forall \trace \in \text{FinTraces}(TS_I): \trace \cdot \text{Traces}(TS_U) \subseteq \bigcup_{\trace \in \text{FinTraces}(TS_I)}\text{Words}(\phi, \psi, \trace).\]
\end{definition}
\vspace{-0.1cm}
To satisfy the universal update condition, the update system needs to be robust against every possible obligation of the initial system.
We explore the model checking and synthesis problems of \liveltl in \Cref{sec:model-checkingandsynthesis}.

\subsection{Obligations}
The impact of the initial system on the update system is declared by the operator $\vDashInitial$ and forms a class of temporal properties.
We investigate this class and build a monitor that traces the open obligations during the execution of a system.
In practice, the explicit update to be performed is unknown during the design of the initial system.
Therefore, one approach to face live updates is keeping track of \textit{open} obligations while the system is executed.
To obtain the expressivity of the obligations possibly enforced by \liveltl, we introduce the \textit{obligation property}.
\begin{definition}[Obligation Property]
    A linear time property $P_{obl}$ over AP is called an \textit{obligation} property if for all words $\sigma \in P_{obl}$ there exists a good prefix, i.e., for every $\sigma \in P_{obl}$ there exists a word $ \sigma[0, m]$ s.t. $\forall x. x \in (2^{AP})^\omega:\ \sigma[0,m] \cdot x \in P_{obl}$. Obligation properties coincide with the class of \textit{co-safety} properties.
\end{definition}
Obligations and co-safety properties describing the same language is a natural outcome of the \liveltl semantics.
To obtain the open obligations with constant cost during runtime, the construction of a monitor tracking the obligations provides a space bounded solution.
The monitor is meant to be constructed simultaneously to the initial system.

\begin{definition}[Obligation Monitor]
        Let $\stripoperator : LTL \rightarrow LTL$ be a function syntactically substituting every $\LTLrelease$ by $\LTLtrue$.
    A deterministic obligation monitor for an LTL formula $\phi$ is the tuple $\obligationmonitor_\phi = (T, t_0, \Upsilon, \af, o)$, where $T = \{ \phi' \mid \omega \in \apstar: \phi' = \af(\phi, \omega)\}$ is the set of states, $t_0 = \stripoperator(\phi)$ is the initial state, $\Upsilon = 2^{AP}$ is the set of directions, $\af$ is the transition function defined over $T$ and $\Upsilon$, and $o(t) = \stripoperator(t)$ is the labeling function.
\end{definition}
Since the state space of $\obligationmonitor_\phi$ corresponds to the state exploration of $\phi$, converting the formulas to obligations is achieved by $\stripoperator$ and stored in the labeling function.
This can be interpreted as the obligations that have to be satisfied by the update system if an update is initiated in this state.
The obligation monitor only tracks states and does not guarantee that every reachable state corresponds to a reachable state of a correct implementation of $\phi$.
We justify this property by assuming $TS_I$ is correct.
\newsavebox{\boxfinally}
\savebox{\boxfinally}{$ \LTLfinally \,$}

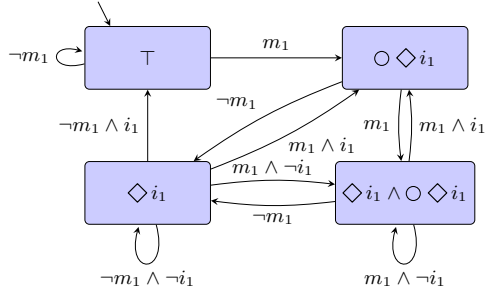
\begin{figure}[t]
\centering
\resizebox{.6\textwidth}{!}{
\begin{tikzpicture}
        \node[draw,fill=blue!20,minimum height=3em,minimum width=6em,rounded corners=2] (N0) at (0,0) {$ \LTLtrue$};
    \node[draw,fill=blue!20,minimum height=3em,minimum width=6em,rounded corners=2] (N3) at (4,0) {$\LTLnext$ \usebox{\boxfinally}$i_1$};
    \node[draw,fill=blue!20,minimum height=3em,minimum width=6em,rounded corners=2] (N4) at (4,-2.1) {\usebox{\boxfinally}$ i_1 \wedge \LTLnext$ \usebox{\boxfinally}$i_1$};
    \node[draw,fill=blue!20,minimum height=3em,minimum width=6em,rounded corners=2] (N5) at (0,-2.1) {\usebox{\boxfinally}$i_1$};

    \node[minimum height=0.2em,minimum width=2.1em,anchor=west] (NN) at (N0.west) {};
    \node[minimum height=0.2em,minimum width=2.1em,anchor=east] (NX) at (N3.east) {};
    \node[circle,inner sep=0pt] (I) at ($ (N0.north west) + (0.2,0.4) $) {};

    \path[->,>=stealth]
    (I)
      edge ($ (N0.north west) + (0.4,0) $)
    (NN)
      edge[loop left] node {$ \neg m_1 $} (N0)
    (NX)
      edge[loop right,opacity=0] node[opacity=0] {\phantom{$ m_{0}  $}} (N3)
    (N0)
      edge node[above] {$ m_{1} $} (N3)
    (N3)
      edge[bend right=7] node [left]{$ m_{1} $} (N4)
      edge[bend right=7] node[above left] {$\neg m_1$}(N5)
    (N4)
      edge[loop below] node {$ m_1 \wedge \neg i_1 $} (N4)
      edge[bend left=7] node[below] {$ \neg m_{1}$} (N5)
      edge[bend right = 7] node [right]{$m_1 \wedge i_1$}(N3)
    (N5)
      edge[loop below] node {$ \neg m_1 \wedge \neg i_1 $} (N5)
      edge[bend left=7] node[above, yshift=-1] {$ m_1 \wedge \neg i_1 $} (N4)
      edge [bend right=7] node[below right, xshift=-3pt, yshift=2pt]{$m_1 \wedge i_1$}(N3)
      edge node[left] {$ \neg m_{1} \wedge i_1 $} (N0)
    ;
  \end{tikzpicture}
  }
\caption{The obligation monitor for $\protect\phi_1$ with one base station.}
\label{fig:minimalOM}
\end{figure}

\begin{example}
\Cref{fig:minimalOM} displays the obligation monitor for $\phi_1 =\LTLglobally(m_1 \rightarrow \LTLnext \LTLeventually i_1)$ of our running example with one base station.
The monitor starts in an obligation free state corresponding to the state before the system is started and contains one direction for every element of $2^{AP}$.
Note that we denote directions symbolically.
Whenever $m_1$ is received on an edge, the obligation $\LTLnext \LTLeventually i_1$ is raised.
From the $\LTLnext \LTLeventually i_1$ state, we differentiate between $m_1$ and $\neg m_1$ leading to another raise of the $\LTLnext \LTLeventually i_1$ obligation together with $\LTLeventually i_1$ or only $\LTLeventually i_1$ respectively.
Returning to the obligation $\LTLtrue$ is only possible if $i_1$ is set to $\LTLtrue$ and $m_1$ is $\LTLfalse$ in the same step.
\end{example}
Note that an offset between initial system and obligation monitor is created.
While transitions of the initial system consider environment inputs and states correspond to system outputs, elements of the state space of the obligation monitor are formulas and the transitions are defined by inputs and outputs combined.
Residing in a state in the obligation monitor can be interpreted as taking a transition in the system and not yet reaching the next state.
\Cref{fig:minimalOM} shows a monitor for a specification, where the implementation is unknown during construction and the obligation monitor over-approximates the reachable states of the implementation.
One can limit the reachable states of the monitor to the paths in the transition system.
Indeed, in regard of completeness, unreachable obligations need to be eliminated from the obligation monitor during verification.

\section{Model Checking and Synthesis}\label{sec:model-checkingandsynthesis}
In this section we solve the problems of model checking live updates and synthesis of live updates, i.e., live synthesis.
We explore finite trace and universal updates for the problems and show the complexity of each result and multiple parameters.

\subsection{Model Checking Live Updates}
Model checking a transition system TS against an LTL formula $\phi$ corresponds to answering the question if TS satisfies $\phi$, i.e., $TS \vDash \phi$. 
For live systems, the evaluation of the update transition system starts with the initial finite execution and switches to the update system afterwards.
Model checking the update system is therefore a language inclusion check of the traces of the transition system combined with $\trace$ against the \liveltl semantics.
\begin{definition}[Model Checking Finite Trace Live Updates]\label{def:modelcheckingfinitetrace}
    Let $TS_U$ be an update system, $\phi$ be an initial specification, $\psi$ be an update specification, and $\trace$ be a finite trace.
    The problem of model checking finite trace live updates is defined as 
    $\trace \cdot Traces(TS_U) \subseteq \text{Words}(\phi,\psi, \trace).$
\end{definition}
The model checking problem can be split into two separate parts, directly identifying the newly introduced conditions for live systems with the operators $\vDashInitial$ and $\vDashUpdate$.
In addition to that, $TS_U$ combined with $\trace$ needs to satisfy the update semantics of \liveltl.
Since both tasks can possibly be performed in isolation of each other, the overhead given by the live update semantics under the assumption of an update system already verified with LTL is an interesting topic but left open for future work.
The complexity of the problem is stated w.r.t. the length of the trace and the combination of initial and update formula:
\begin{theorem}[Complexity in $\phi$, $\psi$, and $\eta$]
    The model checking problem for finite trace live updates is \textsc{PSPACE}-complete in $|\phi| + |\psi|$ and in \textsc{NL} in $\eta \cdot TS_U$.
\end{theorem}
The proof is based on model checking the combination of $\eta$ and $TS_U$.
The universal live update is verified independently of specific initial traces.
The condition is stronger than for finite trace updates, and the number of compatible initial and update systems is smaller.
Given that the context is unknown, the executions starting in the initial state of $TS_U$ need to satisfy every possible open obligation.
Universal updates are relevant if neither the trace nor the obligation monitor are stored and computed respectively.
Given the initial system, model checking universal update compatibility obtains the same complexity as finite trace updates.
\begin{definition}[Model Checking Universal Live Updates]\label{def:modelcheckinguniversal}
   Let $TS_I$ be an initial system, $TS_U$ be an update system, $\phi$ be an initial specification, and $\psi$ be an update specification.
   The problem of model checking universal live updates is defined as
   $\forall \trace \in \text{FinTraces}(TS_I): \trace \cdot Traces(TS) \subseteq \text{Words}(\phi,\psi, \trace).$
\end{definition}
The implicit update points in $TS_I$ allow for the connection of both transition systems and model checking with a linearly increased formula.
\begin{theorem}[Complexity in $\phi + \psi$, and $TS_I \cdot TS_U$]\label{th:modelcheckinguniversal}
    The~model~checking problem for universal live updates is \textsc{PSPACE}-complete in $|\phi|+|\psi|$ and \textsc{NL} in $TS_I \cdot TS_U$.
 \end{theorem}
The complexity results from encoding the live update in the combined transition system $TS_I \cdot TS_U$ and an adapted formula.
Based on the model checking results we introduce live synthesis, the major contribution of this paper.

\subsection{Live Synthesis}\label{sec:livesynthesis}
In this section, we introduce the problem of live synthesis and show the complexity of synthezising live systems.
Synthesis of live updates during the runtime of the initial system promises correct-by-definition updates that can substitute the executed system instantaneously.
In contrast to model checking, the synthesis procedure returns an implementation or \textit{unrealizable}, proving that the finite trace or initial system and initial specification are incompatible with the update specification.
We begin with live updates for an explicit finite trace of the initial system -- the update system needs to react to the explicit context and open obligation.
The definition follows the model checking problem, but searches for a transition system satisfying the live update.
\begin{definition}[Finite Trace Live Synthesis]\label{def:finitetracelivesynthesis}
    Let $\phi$ be an initial specification, $\psi$ be an update specification, and $\trace$ be a finite trace.
    The finite trace live synthesis problem is the computation of a transition system $TS$ s.t. 
    $\trace \cdot \text{Traces}(TS) \subseteq \text{Words}(\phi,\psi, \trace).$
\end{definition}
We additionally call a live update \textit{realizable} if there exists a transition system that satisfies the finite trace live update.
The complexity of the update synthesis is expressed w.r.t. $\phi$
 and $\psi$ and aligns to existing LTL synthesis bounds.
 \begin{theorem}[Complexity in $\phi$ and $\psi$]
    The finite trace live synthesis problem is \textsc{2EXPTIME}-complete in $|\phi|$ and $|\psi|$.
\end{theorem}
The proof is subsumed by the proof of \Cref{th:universallivesynthesis}.
The universal update is again of interest if the context of the live update is unknown.
Synthesizing a transition system that satisfies the universal live update enables the user to plug-in the new system at any time-step without further analysis.
\begin{definition}[Universal Live Synthesis]\label{def:universallivesynthesis}
   Let $\phi$ be an initial specification, $TS_I$ be an initial system, and $\psi$ be an update specification.
   The universal live synthesis problem is the computation of a transition system TS s.t. 
   $\forall \trace \in \text{FinTraces}(TS_I): \trace \cdot \text{Traces}(TS) \subseteq \text{Words}(\phi,\psi, \trace).$
\end{definition}
Again, we call the problem of the existence of a solution realizability.
In general, the universal update obtains a conjunction of double exponentially many conjuncted obligations.
To avoid the expansion of the update system, we combine the parity games of the initial and update system.
Again, the initial formula conducts the impact on the update system and provides the complexity results.

\begin{theorem}[Complexity in $\phi$ and $\psi$]\label{th:universallivesynthesis}
    The universal update synthesis problem is \textsc{2EXPTIME}-complete in $|\phi|$ and $|\psi|$.
\end{theorem}

\begin{proof}[Sketch]
	The hardness proof follows from \Cref{th:liveltlltl}.
	To show the completeness, we sketch the reduction from \liveltl to LTL. 
	Let $\phi'$ be $\phi$ with release formulas limited to the environment AP $update$. 
	We build the parity game of $\phi' \wedge \LTLeventually(update \wedge \LTLnext \psi)$ (cf.~\cite{paritygames}), where $update$ is enforced to only occur once but will eventually hold.	
	We introduce the following changes to the game: The first part of the game ($\phi'$) is restricted to the edges that can be taken in $TS_I$ and all edges are controlled by the environment. 
	Therefore the environment can move arbitrarily in the first game and build any obligation possible.
	Solving the game synthesizes a universal update for the triple $TS_I,\phi, \psi$ regarding the \liveltl semantics.
	Since the reduction is linear in $|\phi|$ and $|\psi|$, we obtain the complexity results from LTL for $\phi$ and $\psi$.
\end{proof}

\section{Case Study}\label{sec:casestudy}
\label{sec:casestudy}
We explore the live update problems on benchmarks from the reactive synthesis competition~\cite{SYNTCOMP} and robot control communities~\cite{RobotSpecifications}.
Our goal is a qualitative analysis of pairs of specifications that can be updated live according to the finite trace live update and the universal live update.
In more detail, we aim to answer the following questions: For specifications that can potentially be updated to each other, does the \liveltl semantics state universally updatable obligations? 
And if not, in how many obligation states is a finite trace update possible?

\sloppy

A prototype for the live synthesis procedure is implemented on top of \textsc{BoSy}~\cite{bosy}, a tool that synthesizes implementations for LTL formulas\footnote{The prototype and experiments are available online at \href{https://github.com/reactive-systems/LiveSynthesisArtifact}{https://github.com/reactive-systems/LiveSynthesisArtifact}.}.
We use \textsc{Spot}~\cite{spot} for LTL formula manipulation and implemented the obligation monitor construction for arbitrary LTL formulas.
For our experiments, the following structure is used: \textsc{BoSy} synthesizes a system for the initial specification which is used to build the obligation monitor.
Therefore, the result of the synthesis query, i.e., a transition system satisfying the formula, is parsed and cut with the obligation monitor to eliminate unreachable states.
Since the result of \textsc{BoSy} may differ per execution, we may obtain different sizes of the obligation monitor for different benchmark runs.
Based on the obligation monitor, we perform explicit trace live synthesis for every monitor state label and universal live synthesis for all monitor states combined. 
Therefore, we build the conjunction of obligation formula and update formula and execute \textsc{BoSy} to check realizability.

\fussy

For the benchmarks in \Cref{sec:benchmarkfamilies}, \Cref{tab:benchmarks} shows multiple results: The number of obligation monitor states built by the initial system and specification, the number of finite trace updates that are realizable, and the result of the universal update.
Despite the finite trace live update stating updates from every possible finite execution of the initial system, we use the state representation of the obligation monitor to symbolically represent every execution.
The runtime in seconds for the update specification without update constraints and the universal update conclude the table.
All experiments were executed on an Intel i7 processor with 2,8 GHz and 16 GB RAM.
\subsection{Benchmark Families}\label{sec:benchmarkfamilies}
The upper part of \Cref{tab:benchmarks} shows the results for live updates from specification patterns introduced by Menghi et. al.~\cite{RobotSpecifications}, where \textbf{Reactivity} implements additional interaction with the environment.
The specifications define the behavior of a robot that is able to travel between n different locations and needs to satisfy different specifications on the way.
Our second set of benchmarks is taken from the annual synthesis competition \textsc{SYNTCOMP}~\cite{SYNTCOMP}.
The results for live updates in the reactive synthesis setting are shown in the lower part of \Cref{tab:benchmarks}. 
\begin{itemize}
	\item \textbf{Visit}, \textbf{Seq. Visit}, and \textbf{Patrolling} enforce the robot to visit every location once, in a sequence, and infinitely often respectively.
	\item \textbf{Reactivity}. The reactivity specification forces the robot to react to an event after two steps at latest by driving to a delineated location, e.g., for refueling. The Reactivity specification can be added to arbitrary specifications.
	\item \textbf{Relay Station}. The running example of this paper. The relay station communicates with $n$ satellites and forwards the message if clients acknowledged.
	\item \textbf{Arbiter}. An arbiter controls the access of multiple clients to a shared resource. It ensures that every request to the resource is eventually granted. We consider three variants of arbiter, a simple arbiter (\textbf{s}) only iterating over grants, a full arbiter (\textbf{f}) only granting access if requested beforehand, and a prioritized arbiter (\textbf{p}) that prioritizes the requests of client 0.
	\item \textbf{ABP}. The alternating bit protocol consists of a receiver \textbf{ABPReceiver} and a transmitter \textbf{ABPTransmitter} specifying the data link layer in the OSI communication network.
	\item \textbf{Load Balancer}. The load balancer distributes workload over $n$ worker.
\end{itemize}
In addition to the specifications, we denote updates with an increased parameter with \textbf{$n \rightarrow n+1$}.
This property is of interest if the parameter may change during the execution, e.g., increasing the number of clients of an arbiter.
\begin{table}[t]
	\centering
	\resizebox{\textwidth}{!}{
	\begin{tabular}{c|c|c|c|c|c|c}
		\hline
		\multicolumn{7}{c}{\textit{Robot Specification Patterns}}\\
		\hline
		\textit{Ben.} & \textit{Update} & \textit{\#OM-States} &  \textit{\#Fin. Trace} & \textit{Universal} & \textit{Time} $\psi$& \textit{Time Univ.}  \\
		\hline
		{Visit} & Seq. Visit &   4 &  4 & \textit{real.} & 0.75&0.75\\
		 & Patrolling & 6 & 6 & \textit{real.} & 0.68 & 0.68\\
		 & Seq. Patrolling & 6 & 6 & \textit{real.} & 0.64 & 0.72\\
		 & Reactivity & 7 & 7 & \textit{real.} & 0.49 & 0.49\\
		 \hline
		 Seq. Visit & Patrolling & 14 & 14 & \textit{real.} & 0.56& 0.59\\
		 & Seq. Patrolling & 16 & 16 & \textit{real.} & 0.57 & 0.59\\
		 & Reactivity & 5 & 5 & \textit{real.} & 0.44& 0.44\\ 
		\hline
		Patrolling  & Ord. Visit  & 6 & 6 & \textit{real.} & 0.61 & 0.67 \\
		
		& Reactivity  & 7 & 7 & \textit{real.} & 0.49 & 0.52\\
		\hline
		\multicolumn{7}{c}{\textit{SYNTCOMP}}\\
		\hline
		Relay Station & 1 $\rightarrow$ 2 & 4 &  4 & \textit{real.} & 16.26 & 17.23\\
		& 2 $\rightarrow$ 1 & 19 & 19 & \textit{real.} & 0.61 & 0.61 \\
		\hline
		 Arbiter & 2f $\rightarrow$ 3f & 11 & 6 & \textit{unreal.} & 5.30 & - \\
		 & 2s $\rightarrow$ 2f & 4 & 2 & \textit{unreal.} & 0.56 & - \\
		 & 2s $\rightarrow$ 4s & 4 & 4 & \textit{real.} & 0.69 & 0.79 \\ 
		 & 2s $\rightarrow$ 2p & 13 & 13 & \textit{real.} & 0.46 & 0.48\\ 
		 & 2f $\rightarrow$ 2p &  10 & 10 & \textit{real.} & 0.45 & 0.52\\ 
		 & 2p $\rightarrow$ 3p &  6& 6 & \textit{real.} & 0.65 & 0.74 \\ 
		\hline
		ABPReceiver & 1 $\rightarrow$ 2 & 5 & 4  & \textit{unreal.} & 0.55 & - \\
		 & 2 $\rightarrow$ 3 & 9 & 3& \textit{unreal.} & 0.43 & - \\
		\hline
		ABPTransmitter & 1 $\rightarrow$ 2 &  5 & 5 & \textit{real.} & 2.70 & 2.82\\
		\hline 
		Load Balancer &  2 $\rightarrow$ 4 & 7 & 7 & \textit{real.} & 0.72 & 0.75 \\
	\end{tabular}}
	\vspace{0.2cm}
\caption{Results of Live Updates for Robot and SYNTCOMP specifications.}
\label{tab:benchmarks}
\vspace{-.8cm}
\end{table}

\subsection{Observations}
Throughout all experiments, the minor runtime overhead of the universal update synthesis shows that the additional cost for live update correctness is feasible.
The robot specifications provide insight of obligations raised during execution.
Since most of the benchmarks obtain the same structural behavior, i.e., the robot visits the locations under some restrictions, the universal live updates are realizable.
Even when adding requests, e.g., the robot has to refuel in two steps after requested, the live update is realizable by satisfying the open obligations after the update.
Changes to the visiting sequence or infinitely often reaching a location with patrolling increases the size of the obligation monitor (\textit{\#OM-States}) but does not lead to unrealizability.
Nevertheless, the sizes of the obligation monitors indicate that tracking the behavior of the system is necessary to obtain the correct obligation.
Altogether, our results show that although robot specifications raise obligations, synthesizing correct live updates is often feasible due to the absence of conflicts between the specifications.  
Most interestingly for the reactive systems benchmarks are arbiter live updates.
Changing a specification to a simple arbiter is realizable since the arbiter does not additionally restrict the behavior.
However, live updates to full arbiter are only possible from some obligation monitor states, shown by the difference of \textit{\#OM-States} and \textit{\#finite trace updates}.
Unrealizability follows from obligation states forcing a grant - an unrequested grant of the update system would be spurious.
Since the prioritized arbiter does not include non-spuriousness, a live update from and to this arbiter is realizable.
The relay station can be universally updated to the one more and one less base stations.
Once computed, the obligations can be satisfied in finite time-steps and synthesizing a solution that reacts to all obligations is possible.

The experiments answer the questions stated at the beginning of this section: Specifications that are meaningful live updates state obligations for the update system, shown by the large number of states of the obligation monitors.
Realizability of the update system depends on the restrictiveness of the specification, even if the universal update is unrealizable, our results show that in all benchmarks some finite trace live updates are realizable.

\section{Related Work}\label{sec:relatedwork}
The necessity of live updates in always-on systems is long known and was introduced as \cite{First,on-the-fly}. 
Dynamic updates for programming languages, e.g.,  in C++~\cite{C++} and Java~\cite{java}, enable developers to update \textit{dynamic classes} during runtime and are called \textit{dynamic software updates} (DSU). 
The proposed frameworks implement functionality and are unable to ensure temporal correctness of the updates.
Live kernel patches received huge attention in the operating system community \cite{operatingsystems2,operatingsystems}, where bug-fixes and features of the kernel can be deployed without reboot.
Recent work in live updates for operating systems achieved real-life implementations, e.g. for Linux \cite{operatingsystems3} and Android \cite{android} kernels. 
Implementations of dynamic updates raised the need for verification:
Following the idea of observability by the user, Hayden et. al.~\cite{DSU} introduce \textit{client-oriented specifications} (CO-specs) to define and verify against client-visible behavior.
Closest to our work are dynamic updates in controller verification and synthesis.
Ghezzi et.al.~\cite{controllersynthesisMSD} introduce a controller synthesis approach based on Modal Sequence Diagrams (MSD).
The update is a synthesized MSD that takes over the execution when a safe state is reached.
While reaching a safe state is also necessary in~\cite{controllersynthesis2}, the authors omit the obligations of the previous system.
Where \cite{controllersynthesisMSD} also relies on the existence of a safe state for the live update,~\cite{controllersynthesisLTL} also proves the reachability of the update state.
Therefore, the condition of the handover between the systems is defined as LTL specification.
The main difference is stating the correctness as LTL formula and not observing the update condition semantically from the initial formula.
\vspace{-0.1cm}

\section{Conclusion}\label{sec:conclusion}
We introduced live synthesis, a synthesis framework for dynamic updates in reactive systems.
We identified \textit{obligations} of a running system as the currently open \textit{co-safety} formulas and defined \liveltl to specify the correct handover between two systems.
The presented obligation monitor enables tracking of obligations during system execution and continuously shows the open obligations.
We explored synthesis and model-checking for two update problems, \textit{finite trace live updates} and \textit{universal update}, which consider full information and zero information of the currently open obligations respectively.
Our case study on robot specifications and reactive synthesis benchmarks show that it is necessary to verify live updates in \textit{always-on} systems and \textit{live synthesis} is able to automatically generate correct update systems if realizable.
We believe that live updates play a crucial role in \textit{high-availability} system verification and can benefit from existing techniques for reactive systems.

\bibliographystyle{splncs04} 
\bibliography{bibliography}

\appendix
\section{Additional Definitions}\label{app:Definitions}
\subsection{LTL Semantics}
\label{app:LTLSemantics}
Let $\sigma = A_0A_1A_2 \ldots$ with $A_i \in 2^{AP}$. $\sigma, i \vDash \varphi$ is defined as:
\begin{align*}
    \sigma, i & \vDash \top \\
    \sigma, i & \nvDash \bot \\
    \sigma, i & \vDash a  &&\text{ iff }\ A_i \vDash a, \text{i.e. } a \in A_i\\
         \sigma, i   &  \vDash    \neg \phi   &&\text{ iff }\ \sigma, i \nvDash \phi  \\
    \sigma, i  & \vDash   \varphi_1 \wedge \varphi_2  &&\text{ iff }\ \sigma, i \vDash \varphi_1 \text{ and}\ \sigma, i \vDash \varphi_2 \\
     \sigma, i  &  \vDash   \LTLnext \varphi   &&\text{ iff }\ {\sigma, i+1 \vDash \varphi} \\
     \sigma, i  &  \vDash   \varphi_1 \LTLuntil \varphi_2   &&\text{ iff }\ \exists j, j \geq i.\ \sigma, j \vDash \varphi_2\ \text{and}\ \forall k,i\leq k < j.\ \sigma, k \vDash \varphi_1 \\
     \sigma, i  &  \vDash \phi_1 \LTLrelease \phi_2  &&\text{ iff }\ \forall j,{j \geq i.}\ \sigma, j \vDash \phi_2\ \text{or} \\
     	&&& \phantom{ iff }\exists k,{k \geq i.}\ (\sigma, k \vDash \phi_1 \wedge \forall l, i \leq l \leq k.\ \sigma, l \vDash \phi_2)
\end{align*}

\subsection{The Function $\af$}\label{app:after}
Let $\varphi$ be an LTL formula and $\nu \in 2^{AP}$. We define the function $\af(\phi, \nu)$ as follows:
\begin{align*}
	\af(\LTLtrue, \nu) &= \LTLtrue\\
	\af(\LTLfalse, \nu) &= \LTLfalse\\	
	\af(a, \nu) &= \twopartdef { \LTLtrue } {a \in \nu} {\LTLfalse} {a \notin \nu}\\
	\af(\neg a, \nu) &= \neg \af(a, \nu)\\
	\af(\phi \wedge \psi, \nu) &= \af(\phi, \nu) \wedge \af(\psi, \nu) \\
	\af(\phi \vee \psi, \nu) &= \af(\phi, \nu) \vee \af(\psi, \nu) \\
	\af(\LTLnext \phi, \nu) &= \varphi\\
	\af(\LTLglobally \phi, \nu) &= \af(\varphi, \nu) \wedge \LTLglobally \phi\\
	\af(\LTLfinally \phi, \nu) &= \af(\varphi, \nu) \vee \LTLfinally \phi\\	
	\af(\phi \LTLuntil \psi, \nu) &= \af(\psi, \nu) \vee (\af(\phi, \nu) \wedge \phi \LTLuntil \psi)\\			
	\af(\phi \LTLrelease \psi, \nu) &= (\af(\phi, \nu) \wedge \af(\psi, \nu)) \vee (\af(\psi, \nu) \wedge \phi \LTLrelease \psi)			
\end{align*}

\subsection{The Function $expand$}\label{app:expand}
Let $\varphi$ be an LTL formula. We define the function $expand(\phi)$ as follows:
\renewcommand{\af}{expand}
\begin{align*}
	\af(\LTLtrue) &= \LTLtrue\\
	\af(\LTLfalse) &= \LTLfalse\\	
	\af(a) &= a\\
	\af(\neg a) &= \neg a\\
	\af(\phi \wedge \psi) &= \af(\phi) \wedge \af(\psi) \\
	\af(\phi \vee \psi) &= \af(\phi) \vee \af(\psi) \\
	\af(\LTLnext \phi) &= \LTLnext \af(\varphi)\\
	\af(\LTLglobally \phi) &= \varphi \wedge \LTLnext \LTLglobally \phi\\
	\af(\LTLfinally \phi) &= \varphi \vee \LTLnext \LTLfinally \phi\\	
	\af(\phi \LTLuntil \psi) &= \psi \vee (\phi \wedge \LTLnext (\phi \LTLuntil \psi))\\			
	\af(\phi \LTLrelease \psi, \nu) &= (\phi \wedge \psi) \vee (\psi \wedge \LTLnext(\phi \LTLrelease \psi))			
\end{align*}

\section{Proofs}\label{app:liveltl:ltl}
\subsection{LiveLTL and LTL}\label{app:liveltl:ltl}
\begin{theorem*}
    \liveltl and LTL are equally expressive.
\end{theorem*}
\begin{proof}
\emph{LTL $\subseteq$ LiveLTL:} To show that $LTL \subseteq \liveltl$, let $\phi \in LTL$ be an arbitrary LTL formula.
One can directly obtain $\text{Words}(\phi)$ with $\text{Words}(\phi',\psi',\eta)$ by instantiating the \liveltl components with $\phi' = \top$, $\psi' = \phi$, and $\eta = \epsilon$.

\emph{\liveltl $\subseteq$ LTL:}
To show that $\liveltl \subseteq LTL$, we first shift the update formula to the end of the finite trace $\trace$ with  $\LTLnext^{|\trace|}(\psi)$.
The evaluation of the formula $\psi$ is thereby delayed to the starting point of the update system.
Secondly, we simulate the handover of the initial formula to the update system:
We use $expand_{|\trace|}$ to unroll the initial formula $\phi$ for $|\trace|$ time-steps.
The resulting formula is equivalent to $\phi$ and the operators $\LTLfinally$ and $\LTLglobally$ only occur inside $|\trace|$ \textit{next} operators.
Following the \liveltl semantics, we limit the influence of $\LTLrelease$: Let $\stripoperator : LTL \rightarrow LTL$ be a function syntactically substituting every $\LTLrelease$ by $\LTLtrue$.
The formula $\stripoperator(\phi')$ on the expanded formula simulates the operator $\vDashInitial$.
To obtain the equivalent LTL formula, we concatenate both formulas with an LTL encoding of the finite trace: $\LTLnext^{|\trace|}(\psi) \wedge \stripoperator(expand_n(\phi)) \wedge LTL(\trace)$, where 
$LTL(\trace) = \bigwedge_{0 \leq i \leq |\trace|-1} \LTLnext^i ((\bigwedge_{a \in \sigma_i} a) \wedge (\bigwedge_{a \in AP\backslash \sigma_i} \neg a))$.
\end{proof}

\subsection{Universal Live Update Model Checking}\label{app:modelcheckinguniversal}
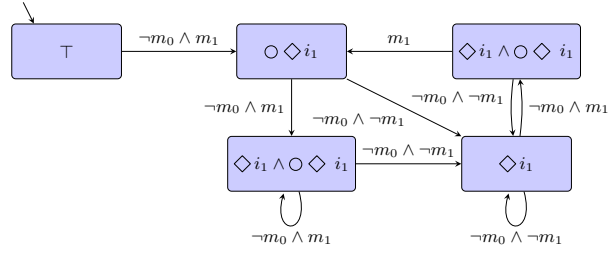
\begin{figure}[t]
\centering
\resizebox{.7\textwidth}{!}{
\begin{tikzpicture}
        \node[draw,fill=blue!20,minimum height=3em,minimum width=6em,rounded corners=2] (N0) at (-4,0) {$ \LTLtrue$};
        \node[draw,fill=blue!20,minimum height=3em,minimum width=6em,rounded corners=2] (N1) at (0,0) {$\LTLnext$ \usebox{\boxfinally}$i_1$};
    \node[draw,fill=blue!20,minimum height=3em,minimum width=6em,rounded corners=2] (N3) at (4,0) {\usebox{\boxfinally}$ i_1 \wedge \LTLnext$ \usebox{\boxfinally} $i_1$};
    \node[draw,fill=blue!20,minimum height=3em,minimum width=6em,rounded corners=2] (N4) at (4,-2) {\usebox{\boxfinally}$i_1$};
    \node[draw,fill=blue!20,minimum height=3em,minimum width=6em,rounded corners=2] (N5) at (0,-2) {\usebox{\boxfinally}$ i_1 \wedge \LTLnext$ \usebox{\boxfinally} $i_1$};

    \node[minimum height=0.2em,minimum width=2.1em,anchor=west] (NN) at (N0.west) {};
    \node[minimum height=0.2em,minimum width=2.1em,anchor=east] (NX) at (N3.east) {};
    \node[circle,inner sep=0pt] (I) at ($ (N0.north west) + (0.2,0.4) $) {};

    \path[->,>=stealth]
    (I)
      edge ($ (N0.north west) + (0.4,0) $)
    (N0)
      edge node[above] {$\neg m_0 \wedge m_1$} (N1)
    (NX)
      edge[loop right,opacity=0] node[opacity=0] {\phantom{$ m_{0}  $}} (N3)
    (N1)
      edge node[below left, xshift = 4pt] {$\neg m_0 \wedge \neg m_1$} (N4)
      edge node[left] {$\neg m_0 \wedge m_1$} (N5)
    (N3)
      edge[bend right=7] node [left, yshift = 5pt]{$ \neg m_0 \wedge \neg m_{1} $} (N4)
      edge node[above] {$m_1$}(N1)
    (N4)
      edge[loop below] node {$ \neg m_0 \wedge \neg m_1 $} (N4)
      edge[bend right = 7] node [right]{$\neg m_0 \wedge m_1$}(N3)
    (N5)
      edge[loop below] node {$ \neg m_0 \wedge m_1 $} (N5)
      edge[] node[above] {$ \neg m_0 \wedge \neg m_1 $} (N4)
    ;
  \end{tikzpicture}
  }
\caption{The obligation monitor for $\protect\phi_1$ with one base station combined with the top two states of the running example.}
\label{fig:escalatorOM}
\end{figure}

\begin{theorem*}[Complexity in $\phi + \psi$, and $TS_I \cdot TS_U$]
    The universal live update model checking problem is \textsc{PSPACE}-complete in $|\phi|+|\psi|$ and \textsc{NL} in $TS_I \cdot TS_U$.
 \end{theorem*}
\begin{proof}[Sketch]
We obtain the lower bound by initializing $\phi$ with $\top$ and $TS_I$ with an empty system.
Model checking with \liveltl semantics is then verifying the update system against the update formula.
For the upper bound, we concatenate both, the specifications and the transition systems.
$TS_I$ and $TS_U$ are connected by a duplicate of every edge in $TS_I$ redirected to the initial state of $TS_U$ and annotated with a newly introduced atomic proposition $update$.
$update$ is controlled by the environment and switching the systems also obtains the inputs from the environment in this step.
The new atomic proposition also encodes the end of the initial system in the formula: Every occurrence of \textit{release} receives the limitation to $update$.
The construction follows the equi-satisfiability proof of \cite{LTLfreduction}.
We therefore obtain a formula that is co-safe if the environment assumption $\LTLeventually update$ holds.
To initialize the formula of the update system, $update$ also spawns $\psi$ which has to hold after starting the update system.
We combine all formulas and obtain one LTL formula and one system that are both linear in the size of the source formulas and transition systems.
Since we base the reduction on known $LTL_f$ results and simple concatenation of transition systems, we can use LTL model checking for \liveltl model checking.
Therefore, we obtain the complexity results from LTL for $|\phi| + |\psi|$ and $TS_I \cdot TS_U$.
\end{proof}

\section{Obligation Monitor Example}\label{app:OM}
We combine \Cref{fig:minimalOM} with the two states at the top of \Cref{fig:relaystation} to compute the exact reachable obligation states, the result is shown in \Cref{fig:escalatorOM}.
For readability, we previously discard all transitions leading out of the two state cycle of the relay station.
Assuming that we start with $\top$ as initial obligation represented by the initial state, we follow the transition $\neg m_0 \wedge m_1$ to the obligation monitor state $\LTLnext \LTLeventually i_1$, the top left state in the implementation.
To give an example for the offset of the monitor and the transition system, if we take the loop $\neg m_0 \wedge m_1$, we stay in the same state but change the obligation: The future system has to satisfy $\LTLeventually i_1 \wedge \LTLnext \LTLeventually i_1$.
The obligation monitor consists of 5 different states with 4 different obligation properties, whereas the implementation only has two states.
The monitor is deterministic and cut to the reachable states and transitions of the implementation.
\end{document}